\documentclass[11pt]{article}

\usepackage[bitstream-charter,greekuppercase=upright,cal=cmcal]{mathdesign}
\usepackage[T1]{fontenc}

\usepackage{amsmath,amsthm,mathabx}

\usepackage{units}

\usepackage[affil-it]{authblk}
\usepackage[english]{babel}

\usepackage[margin=2.2cm]{geometry}
\usepackage{url}
\usepackage{verbatim}
\usepackage{xcolor}
\definecolor{dullmagenta}{rgb}{0.5,0,0.4}
\definecolor{dullblue}{rgb}{0.0,0,0.86}
\usepackage{hyperref}
\hypersetup{colorlinks=true,citecolor=dullblue,linkcolor=dullblue,filecolor=dullblue,urlcolor=dullblue,breaklinks=true,bookmarksnumbered=true, bookmarksopen=true,bookmarksopenlevel=1,pdftitle={Tight lower bound on the error exponent of classical-quantum channels},
pdfsubject={},
pdfauthor={Joseph M. Renes}}
\allowdisplaybreaks
\usepackage{braket}

\newtheorem{theorem}{Theorem}[section]

\newtheorem{lemma}[theorem]{Lemma}

\usepackage{titlesec}

\titleformat*{\section}{\large \bfseries}
\titleformat*{\subsection}{\normalsize\bfseries}

\titlespacing\section{0pt}{12pt plus 4pt minus 2pt}{2pt plus 2pt minus 2pt}

\DeclareFontFamily{U}{bbold}{}
\DeclareFontShape{U}{bbold}{m}{n}
 {
  <-5.5> s*[1.069] bbold5
  <5.5-6.5> s*[1.069] bbold6
  <6.5-7.5> s*[1.069] bbold7
  <7.5-8.5> s*[1.069] bbold8
  <8.5-9.5> s*[1.069] bbold9
  <9.5-11> s*[1.069] bbold12 
  <11-15> s*[1.069] bbold12
  <15-> s*[1.069] bbold17
 }{}

\DeclareRobustCommand{\id}{%
  \text{\usefont{U}{bbold}{m}{n}1}%
}

\usepackage[backend=biber,sorting=none,style=phys-jmr,biblabel=brackets,backref=true,bibencoding=utf8,maxnames=20,eprint=true]{biblatex}
\makeatletter
\pretocmd{\blx@head@bibintoc}{\phantomsection}{}{\ddt}
\makeatother
\DefineBibliographyStrings{english}{%
  backrefpage = {page},
  backrefpages = {pages},
}
\usepackage{csquotes}

\DeclareMathOperator{\tr}{Tr} 
\usepackage{braket}
\newcommand{\ketbra}[1]{\ket{#1}\bra{#1}}

\def\sD{\mathsf{D}}
\def\sH{\mathsf H}
\def\cC{\mathcal C}
\def\cG{\mathcal G}
\def\cN{\mathcal N}
\def\cP{\mathcal P}

\def\cS{\mathcal S}
\def\cX{\mathcal X}
\def\cY{\mathcal Y}
\def\cZ{\mathcal Z}
\def\df{\mathbb F}

\usepackage{todonotes}

\usepackage{titling}

\pretitle{\begin{center} \large \bfseries}
\posttitle{\par\end{center}}
\predate{}
\postdate{}
\preauthor{\begin{center} \normalsize}
\postauthor{\end{center}}
\begin{document}

\title{Tight lower bound on the error exponent of classical-quantum channels}

\author{Joseph M.~Renes}
\affil{\small Institute for Theoretical Physics, ETH Z\"urich, 8093 Z\"urich, Switzerland}

\date{}
\maketitle

\renewcommand{\abstractname}{\vspace{-1.25\baselineskip}} 
\begin{abstract}
A fundamental quantity of interest in Shannon theory, classical or quantum, is the error exponent of a given channel $W$ and rate $R$: the constant $E(W,R)$ which governs the exponential decay of decoding error when using ever larger optimal codes of fixed rate $R$ to communicate over ever more (memoryless) instances of a given channel $W$.
Nearly matching lower and upper bounds are well-known for classical channels. 
Here I show a lower bound on the error exponent of communication over arbitrary classical-quantum (CQ) channels which matches Dalai's  sphere-packing upper bound [\href{https://doi.org/10.1109/TIT.2013.2283794}{IEEE TIT 59, 8027 (2013)}] for rates above a critical value, exactly analogous to the case of classical channels.
This proves a conjecture made by Holevo in his investigation of the problem [\href{https://doi.org/10.1109/18.868501}{IEEE TIT 46, 2256 (2000)}]. 

Unlike the classical case, however, the argument does not proceed via a refined analysis of a suitable decoder, but instead by leveraging a bound by Hayashi on the error exponent of the cryptographic task of  privacy amplification [\href{https://doi.org/10.1007/s00220-014-2174-y}{CMP 333, 335 (2015)}]. 
This bound is then related to the coding problem via tight entropic uncertainty relations and Gallager's method of constructing capacity-achieving parity-check codes for arbitrary channels.  
Along the way, I find a lower bound on the error exponent of the task of compression of classical information relative to quantum side information that matches the sphere-packing upper bound of Cheng et al.\ [\href{https://doi.org/10.1109/TIT.2020.3038517}{IEEE TIT 67, 902 (2021)}]. 
In turn, the polynomial prefactors to the sphere-packing bound found by Cheng et al.\ may be translated to the privacy amplification problem, sharpening a recent result by Li, Yao, and Hayashi [\href{http://doi.org/10.1109/TIT.2022.3217671}{IEEE TIT 69, 1680 (2023)}], at least for linear randomness extractors. 
\end{abstract}
\vspace{1mm}

\section{Introduction}
When communicating over a classical channel $W$ at a rate $R$ below the capacity, a good code family will have a decoding error probability which decays exponentially in the blocklength of the code. 
The optimal decay is characterized by the error exponent, the largest $E(W,R)$ such that the probability of error scales as $2^{-n\,E(W,R)}$ for blocklengths $n\to \infty$. 
This quantity is also known as the reliability function. 
Nearly matching lower and upper bounds on the error exponent for classical channels were first established by the lower bounds of Fano \cite{fano_transmission_1961} and Gallager~\cite{gallager_simple_1965} and the {sphere-packing} upper bound of Shannon, Gallager, and Berlekamp~\cite{shannon_lower_1967}. 
For a channel $W$ mapping a discrete input alphabet $\cX$ to discrete output alphabet $\cY$, these bounds take the form
\begin{equation}
\label{eq:classicalreliabilitybounds}
\max_{s\in[0,1]}E_0(s,W)-sR\,\leq\, E(W,R)\,\leq \,\sup_{s\geq 0} E_0(s,W)-sR\,,
\end{equation}
where $E_0(s,W)\coloneq \max_{P} E_0(s,P,W)$ and $E_0(s,P,W)\coloneq -\log \sum_{y\in \cY} \big(\sum_{x\in \cX} P(x) W(y|x)^{\nicefrac{1}{1+s}}\big)^{1+s}$. 
Here $W(y|x)$ are the channel transition probabilities and $P$ is a probability distribution. 

In this paper I show that the lower bound in \eqref{eq:classicalreliabilitybounds} also applies to arbitrary channels with a classical input but quantum output (CQ channels) for the auxiliary function 
\begin{equation}
\label{eq:quantumauxfunction}
E_0(s,P,W)\coloneq-\log \tr_B[\big(\sum_{x\in \cX} P(x) \varphi_B(x)^{\nicefrac{1}{1+s}}\big)^{1+s}]\,,
\end{equation}
where now the output of the channel $W$ for input $x$ is the quantum state $\varphi_B(x)$ on quantum system $B$. 
Finding a lower bound of this form has been an open question since Burnashev and Holevo initiated the study of CQ error exponents~\cite{burnashev_reliability_1998,holevo_reliability_2000}, defining $E_0$ as in \eqref{eq:quantumauxfunction}, and showing that the lower bound in \eqref{eq:classicalreliabilitybounds} indeed holds for channels with pure state outputs. 
Dalai showed that the upper bound in \eqref{eq:classicalreliabilitybounds} also holds for CQ channels in \cite{dalai_lower_2013}. 
Notably, he also found that the setting of CQ channels encompasses Lov\'asz's bound on the zero-error capacity of a channel~\cite{lovasz_shannon_1979}. 

Unlike more recent lower bounds on the error exponent by Hayashi~\cite{hayashi_error_2007,hayashi_universal_2009}, Dalai~\cite{dalai_note_2017}, and very recently by Beigi and Tomamichel~\cite{beigi_lower_2023-1} (which are not of the form \eqref{eq:classicalreliabilitybounds}), the method employed here does not proceed by analyzing a suitable decoder for a suitably-chosen code. 
It may seem surprising that it would be at all possible to bound $E(W,R)$ without doing so. 
However, it is well known that in the quantum setting the reliability of error correction is related to the secrecy of the task of privacy amplification.
This phenomenon can be expressed in terms of entropic uncertainty relations~\cite{renes_conjectured_2009,berta_uncertainty_2010,coles_uncertainty_2012} and is exploited in security proofs of quantum key distribution in~\cite{lo_unconditional_1999,shor_simple_2000}.
There, known results on quantum error correction are used to ensure the efficacy of privacy amplification, whereas our approach will be to take the opposite route. 

In \cite{renes_duality_2018} I showed a precise relation between the average error probability of communication over CQ channels using linear codes and a particular security parameter of an associated linear extraction function employed for privacy amplification of a certain ``dual'' CQ state. 
The security parameter is measured in terms of the fidelity (or equivalently, purified distance) to the nearest ideal output. 
Fortunately, there exist bounds on the decay of this security parameter in privacy amplification protocols which use linear extractors; particularly relevant here is a bound by Hayashi~\cite{hayashi_precise_2015}. 
Without further modification, though, linear codes only deliver the lowest error exponent for channels with sufficient symmetry.\footnote{An earlier version of this work reported on this case in more detail~\cite{renes_achievable_2023}.} 
However, combining the above results with Gallager's distribution shaping method~\cite[Section 6.2]{gallager_information_1968} and properties of the auxiliary function $E_0$ suffice to give the desired lower bound for arbitrary CQ channels.
Strangely, then, the current tightest random coding argument for CQ channels in the sense of the error exponent actually comes from analysis of privacy amplification! 

The more immediate relation in \cite{renes_duality_2018} is between privacy amplification and compression of classical data relative to quantum side information. 
This gives a lower bound on the error exponent of compression which matches the sphere-packing upper bound found by Cheng et al.~\cite{cheng_non-asymptotic_2021}.  
Furthermore, their sphere-packing bound can be translated into an upper bound on the exponential decay of the security parameter for privacy amplification based on linear extractors, which tightens the results of \cite{li_tight_2023}.
Hence the problem of determining the exponent of the security parameter at very low rates does have a combinatorial nature, as speculated in \cite{li_tight_2023}, as it is inherited from the combinatorial nature of the coding error exponent at low rates.  

The remainder of the paper is structured as follows. 
The following section provides the necessary mathematical setup, and then the proof of the main result is given in Section~\ref{sec:mainresult}. 
The main result depends on a more general statement, which is the subject of Section~\ref{sec:lemma}. 
Section~\ref{sec:compression} then examines its implications for compression with side information and privacy amplification.

\section{Mathematical setup}
\subsection{Entropies}
To establish these results first requires some preliminary mathematical setup. 
Recall the Umegaki relative entropy of two quantum states $\rho$ and $\sigma$ is given by $D(\rho,\sigma)\coloneq\tr[\rho(\log \rho-\log \sigma)]$. 
Here, and throughout, $\log$ denotes the base two logarithm. 
We require two versions of the R\'enyi relative entropy, one by Petz and the other the minimal version in a certain sense (see Tomamichel \cite{tomamichel_quantum_2016-1} for an overview). 
The Petz version of the R\'enyi relative entropy of order $\alpha\in \mathbb R$ is 
\begin{equation}
\bar D_\alpha(\rho,\sigma)\coloneq\tfrac{1}{\alpha-1}\log \,\tr[\rho^\alpha\sigma^{1-\alpha}]\,,
\end{equation}
while the minimal (or ``sandwiched'') version is 
\begin{equation}
\widetilde D_\alpha(\rho,\sigma)\coloneq\tfrac{1}{\alpha-1}\log \,\tr[(\sigma^{\tfrac{1-\alpha}{2\alpha}}\rho\sigma^{\tfrac{1-\alpha}{2\alpha}})^\alpha]\,.
\end{equation}
Observe that $\widetilde D_{\nicefrac 12}(\rho,\sigma)=-\log F(\rho,\sigma)^2$, where $F(\rho,\sigma)\coloneq \|\rho^{\nicefrac 12}\sigma^{\nicefrac 12}\|_1$ is the fidelity. 
It is known that $\lim_{\alpha\to 1}\widetilde D_\alpha(\rho,\sigma)=D(\rho,\sigma)$ and that $\alpha\mapsto \widetilde D_\alpha(\rho,\sigma)$ is monotonically increasing~\cite{muller-lennert_quantum_2013} (in fact, the same holds for $\bar D_\alpha$). 
Thus, we immediately have the bound
\begin{equation}
\label{eq:fidelityrelentropy}
F(\rho,\sigma)^2\geq 2^{-D(\rho,\sigma)}\,.
\end{equation}

From these two relative entropies we can define several conditional entropies of a bipartite state $\rho_{AB}$ which will be of use to us. For notational clarity, let $\cS_B$ be the set of quantum states on system $B$. The conditional entropies are 
\begin{align}
\bar H_\alpha^\uparrow(A|B)_\rho &\coloneq\max_{\sigma_B\in \cS_B}[-\bar D_\alpha(\rho_{AB},\id_A\otimes \sigma_B)]\,,\label{eq:HSibson}\\
\widetilde H_\alpha^\downarrow(A|B)_\rho &\coloneq-\widetilde D_\alpha(\rho_{AB},\id_A\otimes \rho_B)\,\\
\widetilde H_\alpha^\uparrow(A|B)_\rho &\coloneq\max_{\sigma_B\in \cS_B}[-\widetilde D_\alpha(\rho_{AB},\id_A\otimes \sigma_B)]\,.
\end{align}
(We will not need $\bar H_\alpha^\downarrow$.)
The optimal $\sigma_B^\star$ in \eqref{eq:HSibson} is known from the quantum Sibson identity~\cite{sharma_fundamental_2013}. 
It states that for all $\alpha\geq 0$, 
\begin{equation}
\label{eq:sibson}
\sigma_B^\star = \frac{(\tr_A[\rho_{AB}^\alpha])^{\nicefrac{1}{\alpha}}}{\tr[(\tr_A[\rho_{AB}^\alpha])^{\nicefrac{1}{\alpha}}]}\,.
\end{equation}
This form ensures that the auxiliary reliability function of a CQ channel and uniform input distribution can be related to the conditional Petz R\'enyi entropy of a suitable state. 
Fix a channel $W_{B|Z}$ which maps $z\in \cZ$ to $\varphi_B(z)\in \cS_B$. 
Consider the case of the uniform distribution $Q$ and define the state 
$\rho_{ZB}=\frac{1}{|\cZ|}\sum_{z\in \cZ}\ketbra{z}_Z\otimes \varphi_B(z)$. 
Then, for $E_0$ defined in \eqref{eq:quantumauxfunction}, 
\begin{equation}
\label{eq:renyireliability}
\bar{H}^\uparrow_\alpha(Z|B)_\rho=\log |\cZ|-\tfrac{\alpha}{1-\alpha}E_0(\tfrac{1-\alpha}\alpha,Q,W)\,.
\end{equation}
Equivalently, $E_0(s,Q,W)=s(\log |\cZ|-\bar H_{\nicefrac1{1+s}}^\uparrow(Z|B)_\rho)$.
To see this, it is convenient to let $\theta_B=\sum_{z\in \cZ}\frac{1}{|\cZ|^\alpha}\varphi_B(z)^\alpha$. 
Then according to \eqref{eq:sibson}, we have $\sigma_B^\star=\theta_B^{\nicefrac 1\alpha}/\tr[\theta_B^{\nicefrac1\alpha}]$.
From there a tedious calculation reveals $\bar{H}^\uparrow_\alpha(Z|B)_\rho=-\tfrac{\alpha}{\alpha-1}\log \tr[\theta_B^{\nicefrac 1\alpha}]$, which gives \eqref{eq:renyireliability} after putting the proper power of $\frac1{|\cZ|}$ inside the trace. 

For a CQ state $\rho_{ZB}=\sum_z P(z)\ketbra{z}_Z\otimes \varphi_B(z)$ with an arbitrary prior probability distribution $P$, the optimal probability of determining $Z$ by measuring $B$ is given by
\begin{equation}
P_{\mathrm{guess}}(Z|B)_\rho\coloneq \max_{\Lambda}\sum_{z\in \cZ}P(z)\tr[\Lambda_B(z)\varphi_B(z)]\,,
\end{equation}
where the optimization is over all POVMs $\Lambda$ with elements $\Lambda_B(z)$.
This quantity is directly related to the min-entropy $H_{\min}(A|B)_\psi\coloneq \max_{\lambda\in \mathbb R,\sigma_B\in \cS_B}\{\lambda:\psi_{AB}\leq 2^{-\lambda}\id_A\otimes \sigma_B\}$ by $P_{\mathrm{guess}}(Z|B)_\rho=2^{-H_{\min}(Z|B)_\rho}$~\cite{konig_operational_2009}. 
Moreover, the min-entropy is in fact one of the R\'enyi conditional entropies: $H_{\min}(A|B)_\psi=\widetilde H^\uparrow_\infty(A|B)_\psi$~\cite{muller-lennert_quantum_2013}. 
We will be interested in the probability of guessing the input of a CQ channel given its output, under the assumption that the input is uniformly distributed. 
Calling the channel $W_{B|X}$, we denote this probability by $P_{\mathrm{guess}}(W)$; it is equal to $P_{\textnormal{guess}}(X|B)_\rho$ for $\rho_{XB}=\tfrac1{|\cX|}\sum_x \ketbra{x}_X\otimes \varphi_B(x)$. 
We will also make use of $P_{\textnormal{error}}(W)\coloneq 1-P_{\textnormal{guess}}(W)$.

\subsection{Entropy duality}
\label{sec:entropyduality}
The R\'enyi conditional entropies satisfy a number of interesting and useful duality relations. For our purposes the following two are important.
For any pure state $\rho_{ABC}$~\cite{tomamichel_relating_2014},\cite{beigi_sandwiched_2013,muller-lennert_quantum_2013}
\begin{align}
\bar H_\alpha^\uparrow(A|B)_\rho+\widetilde H_{1/\alpha}^\downarrow(A|C)_\rho&=0\qquad \alpha\in[0,\infty]\,,\\
\widetilde H_\alpha^\uparrow(A|B)_\rho+\widetilde H_{\alpha/(2\alpha-1)}^\uparrow(A|C)_\rho&=0\qquad \alpha\in[\tfrac12,\infty]\,.
\end{align}
The conditional von Neumann entropy itself is self-dual. 
Entropy duality implies entropic uncertainty relations between conjugate observables~\cite{renes_conjectured_2009,berta_uncertainty_2010,coles_uncertainty_2012}. 
For a $d$-level quantum system $A$, let $\{\ket{z}\}_{z\in \mathbb Z_d}$ and $\{\ket{\widetilde x}\}_{x\in \mathbb Z_d}$ be two orthonormal bases of $\mathcal H_A$ such that $|\braket{\widetilde x|z}|^2=\tfrac1{d_A}$ for all $x,z\in \mathbb Z_d$.
Abusing notation somewhat, we also denote the random variables associated with outcomes of measuring in either bases by $Z_A$ and $X_A$, respectively. Denoting by $\cP_A$ ($\widetilde \cP_A$) the quantum channel which pinches $A$ in the $Z$ ($X$) basis, we write e.g.\ $\bar H^\uparrow_\alpha(Z_A|B)_\rho$ for $\bar H^\uparrow_\alpha(A|B)_{\cP_{A}[\rho_{AB}]}$. 
Then, for any quantum state $\rho_{ABC}$, we have
\begin{align}
\bar H_\alpha^\uparrow(Z_A|B)_\rho+\widetilde H_{1/\alpha}^\downarrow(X_A|C)_\rho&\geq \log d\,,\\
\widetilde H_\alpha^\uparrow(Z_A|B)_\rho+\widetilde H_{\alpha/(2\alpha-1)}^\uparrow(X_A|C)_\rho&\geq \log d\,.
\end{align}
In fact, these inequalities are saturated for certain quantum states, as detailed in \cite{renes_duality_2018}. 
For completeness, a concise self-contained discussion is given in Appendix~\ref{sec:duality}. 
In particular, for pure states $\psi_{AA'BC}$ with $A'\simeq A$ which are invariant under the action of the projector $\Pi_{AA'}\coloneq \sum_z \ketbra{z}_A\otimes \ketbra{z}_{A'}$, it holds that 
\begin{align}
\label{eq:entropyduality1}
\bar H_\alpha^\uparrow(Z_A|B)_\psi+\widetilde H_{1/\alpha}^\downarrow(X_A|A'C)_\psi&= \log d_A\,,\\
\widetilde H_\alpha^\uparrow(Z_A|B)_\psi+\widetilde H_{\alpha/(2\alpha-1)}^\uparrow(X_A|A'C)_\psi&=\log d_A\,.\label{eq:entropyduality2}
\end{align}
Once again, this equality also holds for the conditional von Neumann entropy. 
For the min-entropy, corresponding to $\alpha=\infty$ in the latter equation, the dual entropy is $\alpha=\nicefrac 12$, which is related to the fidelity. 
Using the relation of min-entropy to guessing probability, \eqref{eq:entropyduality2} implies that for pure states $\psi_{AA'BC}$,
\begin{equation}
\label{eq:PguessF}
P_{\mathrm{guess}}(Z_A|B)_\psi=\max_{\sigma_{A'C}\in \cS_{A'C}} F(\widetilde{\cP}_A[\psi_{AA'C}],\pi_A\otimes \sigma_{A'C})^2\,.
\end{equation}

\subsection{Duality of linear functions}
\label{sec:dualityfunctions}
Conjugate bases with additional algebraic structure have additional duality properties. 
The usual example is when the two bases are related by a discrete Fourier transform. 
Specifically, consider a system $A$ whose dimension $q$ is a prime integer and define $\ket{\widetilde x}=\tfrac1{\sqrt{q}}\sum_{z\in \df_q}\omega^{xz}\ket{z}$, where $\omega=e^{2\pi i/q}$. Then $|\braket{z|\widetilde x}|^2=\frac1q$, so the bases are indeed conjugate. 
Next consider $n$ copies of $A$, denoted $A^n$, which is shorthand for $A_1A_2\dots A_n$. 
The vectors $\{\ket{z^n}\coloneq\ket{z_1}\otimes \ket{z_2}\otimes \cdots \otimes \ket{z_n}\}_{z^n\in \df_q^n}$ form a basis, as do the vectors $\{\ket{\widetilde x^n}\}$. 
The vectors $\ket{\widetilde x^n}$ can be expressed as $\ket{\widetilde x^n}=\tfrac1{\sqrt{q^n}}\sum_{z^n\in \df_q^n}\omega^{x^n\cdot z^n}\ket{z^n}$, where $x^n\cdot z^n=\sum_{k=1}^n x_k z_k$. 
Therefore the bases are also conjugate, since $|\braket{z^n|\widetilde x^n}|^2=\tfrac1{q^n}$. 

Importantly, a unitary which implements an invertible function in one basis also implements an invertible function in the Fourier conjugate basis. 
More concretely, suppose that $f$ is an invertible linear map from $\df_q^n$ to itself and let $U$ be the unitary which maps $\ket{z^n}$ to $\ket{f(z^n)}$. 
Being linear, $f$ has a matrix representation as $f:z^n\mapsto Mz^n$, and so we have 
\begin{equation}
\begin{aligned}
U\ket{\tilde x^n}
&=\sum_{z^n\in \df_q^n} \ket{f(z^n)}\braket{z^n|\tilde x^n}
=\tfrac{1}{q^{n/2}}\sum_{z^n\in \df_q^n}\omega^{x^n\cdot z^n}\ket{Mz^n}
=\tfrac{1}{q^{n/2}}\sum_{z^n\in \df_q^n}\omega^{x^n\cdot M^{-1}(z^n)}\ket{z^n}\\
&=\tfrac{1}{q^{n/2}}\sum_{z^n\in \df_q^n}\omega^{(M^{-1})^Tx^n\cdot z^n}\ket{z^n}
=\ket{\widetilde {(M^{-1})^Tx^n}}.
\end{aligned}
\end{equation}
Therefore, the action of $U$ in the Fourier conjugate basis is the linear map $g:x^n\mapsto (M^{-1})^Tx^n$.

The above results also extend to surjective linear maps $\hat f: \df_q^n\to \df_q^k$ for $k<n$, in that the action of $\hat f$ on the basis $\ket{z^n}$ can also be understood as implementing a surjective linear map $\hat g:\df_q^n\to \df_q^k$ on the basis $\ket{\widetilde x^n}$. 
This follows straightforwardly because $\hat f$ can be extended to the invertible function $f$ with action $f:z^n \mapsto \hat f(z^n)\oplus \check f(z^n)$, where $\check f:\df_q^n\to \df_q^{n-k}$ is any function to the cosets of the kernel of $\hat f$ in $\df_q^n$, i.e.\  $\df_q^n/\textnormal{ker}(\hat f)$. 
In terms of the representative $M$ of $f$, $\hat f$ defines a rectangular matrix with linearly independent rows. 
Any such matrix can be extended with additional linearly independent rows forming a basis of $\df_q^n$, and the resulting matrix defines $M$. 
The first $k$ rows of $M$ are again $\hat f$, and then $\hat g$ is simply the first $k$ rows of $(M^{-1})^T$. 
Calling this matrix $G$, the action of $\hat g$ is $x^n\mapsto Gx^n$.
Both functions $\hat f$ and $\hat g$ are implemented by the unitary $U$. 
In this context, it is convenient to regard $U$ as a map from $A^n$ to the compound system $\hat A\check A$, where $\hat A$ is a system of $k$ qubits and $\check A$ is a system of $n-k$ qubits. 
Indeed, $\hat A$ is simply the first $k$ qubits of $A^n$, but the different labels help indicate which system is which. 

Surjective linear maps are closely related to linear and affine error-correcting codes. 
The approach we take here is to regard $\check f$ as the function which returns the syndrome of the input, i.e.\ the lower $n-k$ rows of $M$ are the parity-check matrix $H$ of the code. 
The rows of the corresponding $k\times n$ generator matrix are codewords of the code. 
By construction, the generator matrix of a linear code is simply the matrix $G$ defined above, because it satisfies $HG^T=0$. 
In contrast to $\hat g$, though, the action of the encoding operation is $m^k\mapsto m^k G$ (regarding $m^k$ as a row vector). 
For an affine code, the parity-checks of the codewords are not zero, but take some other value, say $s^{n-k}$.
Affine codes are simply cosets of linear codes in $\df_q^n$, and can be encoded by $m^k\mapsto m^kG + v^n$, where $v^n$ is a suitable coset leader satisfying $Hv^n=s^{n-k}$. 

Of particular relevance will be generator matrices $G$ which have the form $G=\begin{pmatrix} \id_k & T\end{pmatrix}$, where $T$ is a $k\times (n-k)$ Toeplitz matrix, meaning all entries along a given diagonal are identical. 
We will refer to the collection of $G$ of this form as modified Toeplitz matrices. 
The resulting codewords are systematic encodings of the message, along with a kind of convolution of the message.
Though here the convolution operation potentially involves the entire message, not just a limited portion of it. 

For $\hat g$ of modified Toeplitz form with Toeplitz matrix $T$, one choice of $\hat f$ and $\check f$ is given by $\hat f$ simply mapping $z^n$ to its first $k$ elements (matrix representation $\begin{pmatrix} \id_k & 0\end{pmatrix}$) and $\check f$ having matrix representation $\begin{pmatrix} -T & \id_{n-k}\end{pmatrix}$. 
This follows because $M=
\begin{pmatrix}
\id_{k} & T\\ 0 & \id_{n-k}
\end{pmatrix}$
is invertible and its inverse has the same form, but with $T$ replaced by $-T$.

\section{Main result: Lower bound on the error exponent of CQ channels}
\label{sec:mainresult}
To reliably transmit information of $n$ i.i.d.\ uses of a CQ channel $W_{B|Y}$ will generally require the use of a classical code $\cC$ mapping the message space $\mathcal M$ to $\mathcal Y^{\times n}$. 
This combination forms a new channel which we denote by $W^{\otimes n}\circ \cC$. 
Here we have overloaded notation somewhat to use $\cC$ to describe the code and the classical channel which implements the encoding function.

We are interested in the error probability at blocklength $n$ of the optimal code $\cC$: 
\begin{equation}
P_{\textnormal{error,min}}(W^{\otimes n})\coloneq \min_{\cC}P_{\textnormal{error}}(W^{\otimes n}\circ\cC)\,.
\end{equation}
With this quantity we can define the error exponent of the channel as 
\begin{equation}
E(W,R)\coloneq \limsup_{n\to \infty} [-\tfrac1n\log P_{\textnormal{error,min}}(W^{\otimes n})]\,.
\end{equation}
Now we can state and prove our main result. 
\begin{theorem}
\label{thm:main}
For an arbitrary CQ channel $W_{B|Y}$ whose input $Y$ comes from a discrete alphabet $\cY$ and whose output is a density operator on a quantum system $B$ with finite-dimensional state space, and any rate $R\geq 0$,
\begin{equation}
\label{eq:explowerbound}
E(W,R)\geq \sup_{s\in [0,1]} E_0(s,W)-sR\,.
\end{equation}
\end{theorem} 

\begin{proof}
The proof has two main ingredients. The first is a bound on the error probability of channels when using affine codes, and the second is Gallager's distribution shaping method for mimicking arbitrary channel input distributions using linear codes~\cite[Section 6.2]{gallager_information_1968}. 
Let us begin with the latter. 

Given a distribution $P$ over an alphabet $\mathcal Y=\mathbb Z_r$ for some $r\in \mathbb N$, consider a quantization of $P$ to $q>r$ values, with $q$ prime: a distribution $P'$ such that $P'(y)=w_y/q$ for $w_y\in \mathbb N$ and $\sum_{y\in \cY} w_y=q$. 
According to \cite[Proposition 2]{bocherer_optimal_2016}, for every $P$ there exists a quantized version $P'$ such that the variational distance satisfies $\delta(P,P')\leq \nicefrac r{4q}$. 
Let $\cG_{Y|Z}$ be the classical channel that implements the function $b:\df_q\to \mathcal Y$ which maps the first $w_0$ values of $z\in \df_q$ to 0, the next $w_1$ values of $z\in \df_q$ to $1$, and so on. 

Now take $P$ to be an optimizer in $E_0(s,W)\coloneq\max_{P}E_0(s,P,W)$ and consider the channel $W'_{B|Z}=W_{B|Y}\circ \cG_{Y|Z}$. 
Crucially, its auxiliary function under the uniform distribution $Q$ is precisely that of $W$ under $P'$:
\begin{align}
E_0(s,Q,W')
&=-\log \tr[\big(\sum_z \tfrac1q \varphi_{B}(b(z))^{\nicefrac{1}{1+s}}\big)^{1+s}]
=-\log \tr[\big(\sum_y \tfrac{w_y}{q} \varphi_{B}(y)^{\nicefrac{1}{1+s}}\big)^{1+s}]=E_0(s,P',W)\,.
\end{align}
Moreover, $E_0$ is additive, so that $E_0(s,Q^{\otimes n},(W')^{\otimes n})=nE_0(s,P',W)$. 
Having chosen $q$ prime, we can apply the following to bound the error of coding over $\hat W=(W')^{\otimes n}$.
\begin{lemma}
\label{lem:oneshoterror}
For an arbitrary CQ channel $\hat W_{B|Z}$ with input alphabet $\df_q^n$ for $q$ a prime and whose outputs $\varphi_B(z)$ are density operators of a finite-dimensional quantum system $B$, define the state $\rho_{ZB}=\sum_{z\in \df_q^n} Q(z)\ketbra{z}_Z\otimes \varphi_B(z)$, where $Q$ is the uniform distribution. 
Then there exists an affine code $\cC$ of size $|\cC|=q^{m}$ whose generator matrix is a modified Toeplitz matrix such that  
for any $s\in [0,1]$, 
\begin{equation}
\label{eq:oneshoterror}
-\log P_{\mathrm{error}}(\hat W\circ \mathcal C)\geq E_0(s,Q,\hat W)-s\log |\cC|-s\log |\mathrm{spec}(\rho_{ZB})| -\log \tfrac1s \,.
\end{equation}
\end{lemma}
The lemma will be proven in the following section. 
Note that the bound holds for arbitrary $s\in [0,1]$ and code rate $R\coloneq \tfrac1n \log |\cC|$ (such that $|\cC|=q^m$ for integer $m$). 
For the problem at hand, the associated state $\rho_{Z^nB^n}$ is simply $\rho_{Z^nB^n}=(\rho_{ZB})^{\otimes n}$ for $\rho_{ZB}=\frac1q \sum_{z\in \df_q}\ketbra{z}\otimes \varphi_B\big(b(z)\big)$. 
By the usual type-counting arguments (e.g.~\cite[Theorem 11.1.1]{cover_elements_2006}), we have 
\begin{equation}
\log |\textnormal{spec}(\rho_{Z^nB^n})|\leq |\textnormal{spec}(\rho_{ZB})|\,\log (n+1)\leq r|B|\log (n+1)\,,
\end{equation}
where the second inequality holds due to the form of $\rho_{ZB}$. 
Nominally we would have $|\textnormal{spec}(\rho_{ZB})|\leq q|B|\log (n+1)$, but there are many degeneracies in $\rho_{ZB}$ by construction. 
Thus, for any $s\in [0,1]$, we have 
\begin{equation}
-\log P_{\mathrm{error}}((W\circ \cG)^{\otimes n}\circ \cC)\geq nE_0(s,P',W)-nsR-sr|B|\log (n+1) -\log \tfrac1s\,.
\end{equation}

The function $P\mapsto E_0(s,P,W)$ is continuous in $P$ for all $s\in [0,1]$ (see e.g.\ \cite[Proposition 9.1]{cheng_error_2018}). 
Therefore $E_0(s,P',W)-E_0(s,P,W)\leq \Delta$ for some quantity $\Delta$ which goes to zero when $\delta(P,P')\to 0$. 
Using the aforementioned bound $\delta(P,P')\leq \nicefrac r{4q}$, we can ensure that $\delta\to 0$ as $n\to \infty$ simply by taking $q=O(n)$ or even $q=O(\log n)$. 
Dividing through by $n$ gives
\begin{equation}
-\tfrac1n\log P_{\mathrm{error}}((W\circ \cG)^{\otimes n}\circ \cC)\geq E_0(s,P,W)-sR - \Delta-sr|B|\tfrac{\log (n+1)}n -\tfrac1n\log \tfrac1s\,.
\end{equation}
Given our choice of $P$, in the limit we obtain, for any $s\in[0,1]$ and any choice of $R$, 
\begin{equation}
\lim_{n\to\infty}[-\tfrac1n\log P_{\mathrm{error}}((W\circ \cG)^{\otimes n}\circ \cC)]\geq E_0(s,W)-sR\,.
\end{equation}
Now interpret $\cG^{\otimes n}\circ \cC$ as the coding scheme for $W^{\otimes n}$. 
Optimizing over $s$ gives the desired statement.
\end{proof}

We expect the bound to only be useful for $R<C(W)$, where $C(W)=\max_P I(Y:B)$ is the capacity of $W$. 
This follows from the bound $E_0(s,P,W)\leq sI(Y{:}B)$ established by Holevo~\cite[Proposition 1]{holevo_reliability_2000}. 
For any $R>C(W)$, the quantity $\sup_{s\in [0,1],P} E_0(s,P,W)-sR\leq \sup_{s\in [0,1],P} s(I(Y{:}B)-R)\leq 0$. 
Hence, when $R>C(W)$ the bound in \eqref{eq:explowerbound} is nonpositive.

Theorem~\ref{thm:main} reproduces the results found by Burnashev and Holevo \cite{burnashev_reliability_1998} for CQ channels with pure state outputs, as well as an earlier version of this work restricted to channels with suitable symmetries~\cite{renes_achievable_2023}. 
Moreover, \eqref{eq:explowerbound} compares favorably with Dalai's sphere packing upper bound~\cite[Theorem 5]{dalai_lower_2013},
\begin{equation}
E(W,R)\,\leq\, \sup_{s\geq 0} \,E_0(s,W)-sR\,.
\end{equation} 
Together, the upper and lower bounds reduce to the known results for classical channels. 

The proof of Lemma~\ref{lem:oneshoterror} makes use of a random-coding type argument in which the codes $\cC$ can be chosen to be modified Toeplitz matrices described at the end of Section~\ref{sec:dualityfunctions}. 
Therefore we establish that the lower bound in \eqref{eq:classicalreliabilitybounds} can be achieved not just by random codes but by highly structured codes.

\section{One-shot error probability bound}
\label{sec:lemma}
It remains to prove Lemma~\ref{lem:oneshoterror}. 
In this section we will establish a somewhat more general result. 
Let $\ket{\varphi(z)}_{BC}$ be a purfication of the channel output state $\varphi_B(z)$. 
Given any probability mass function $P$ over $\cZ$, define the pure state 
\begin{equation}
\ket{\psi'}_{A^nA'^nBC}=\sum_{z^n\in \df_q^n}\sqrt{P(z^n)}\ket{z^n}_{A^n}\ket{z^n}_{A'^n}\ket{\varphi(z^n)}_{BC}\,.
\end{equation}
For an invertible linear function $f$ on $\df_q^n$, let $U_A$ be the associated unitary operator as in Section~\ref{sec:dualityfunctions} and set $\ket{\psi}_{A^nA'^nBC}=U_A\ket{\psi'}_{A^nA'^nBC}$. 
Define $\hat A$ and $\check A$ so that $A^n\simeq \hat A\check A$. 
Then we have
\begin{lemma}
\label{lem:basicerrorbound}
For the state $\ket{\psi}_{A^nA'^nBC}$ just defined and $\hat f:\df_q^n\to \df_q^k$ with matrix representation $\begin{pmatrix}\id_k & 0\end{pmatrix}$, there exists a surjective function $\check f:\df_q^n\to \df_q^n/\textnormal{ker}(\hat f)$ and $s\in[0,1]$,
\begin{equation}
-\log P_{\textnormal{error}}(Z_{\hat A}|BZ_{\check A})_\psi \geq s(n-k)\log q-s\bar H^\uparrow_{\nicefrac 1{1+s}}(Z_{A^n}|B)_\psi-s\log |\textnormal{spec}(\psi_{A^nB})| -\log \tfrac1s\,.
\end{equation}
\end{lemma}
\begin{proof}
Consider an arbitrary surjective linear $\hat f$ and $\check f$ pair. 
By \eqref{eq:PguessF} we have, for any conjugate basis of $\hat A$, 
\begin{equation}
P_{\textnormal{guess}}(Z_{\hat A}|BZ_{\check A})_\psi=\max_{\sigma_{C}\in \cS_C}F(\widetilde P_{\hat A}[\psi_{\hat AA'^nC}],\pi_{\hat A}\otimes \psi_{A'^nC})^2\,.
\end{equation}
Nominally, we ought to have $B\check A$ instead of $BZ_{\check A}$ in the conditional of the guessing probability. 
However, due to the form of the state, $\check A$ is diagonal in the $\ket{\check z}$ basis once $A'^n$ is traced out. 
In light of \eqref{eq:fidelityrelentropy}, we have 
\begin{equation}
P_{\textnormal{guess}}(Z_{\hat A}|BZ_{\check A})_\psi \geq 2^{-D(\widetilde P_{\hat A}[\psi_{\hat AA'^nC}],\pi_{\hat A}\otimes \psi_{A'^nC})}\,.
\end{equation}
And because $1-2^{-x}\leq x$ for $x\geq 0$, we obtain
\begin{equation}
\label{eq:decouplingstart}
P_{\textnormal{error}}(Z_{\hat A}|BZ_{\check A})_\psi\leq D(\widetilde P_{\hat A}[\psi_{\hat AA'^nC}],\pi_{\hat A}\otimes \psi_{A'^nC})\,.
\end{equation}

The righthand side of \eqref{eq:decouplingstart} is a quantification of how decoupled $\hat A$ is from $A'^nC$ in the state $\psi_{\hat AA'^nC}$. 
For it to be useful, though, we need to relate it to quantities involving the state on $A^nA'^nC$. 
Fortunately, this is possible if we choose $\ket{\widetilde x^n}$ to be the Fourier conjugate basis, as detailed in Section~\ref{sec:dualityfunctions}. 
Then we can make use of \cite[Theorem 1]{hayashi_precise_2015}, reproduced in Appendix~\ref{sec:hayashibound}. 
It implies that there exists a surjective function $\hat g:\df_q^n\to \df_q^{k}$ of modified Toeplitz form such that for all $s\in [0,1]$,
\begin{equation}
D(\widetilde P_{\hat A}[\psi_{\hat AA'^nC}],\pi_{\hat A}\otimes \psi_{A'^nC})\leq \tfrac 1s|\text{spec}(\psi_{A'^nC})|^s (q^{k})^s\,2^{-s\widetilde H^\downarrow_{1+s}(X_{A^n}|A'^nC)_\psi}\,.
\end{equation}
Setting $\hat f$ and $\check f$ from $\hat g$ as in Section~\ref{sec:dualityfunctions} and using the fact that the spectra of $\psi_{A^nB}$ and $\psi_{A'^nC}$ are equal since the overall state on $A^nA'^nBC$ is pure, we obtain
\begin{equation}
-\log P_{\textnormal{error}}(Z_{\hat A}|BZ_{\check A})_\psi \geq s\widetilde H^\downarrow_{1+s}(X_{A^n}|A'^nC)_\psi-sk\log q-s\log |\text{spec}(\psi_{A^nB})| -\log \tfrac1s\,.
\end{equation}
From \eqref{eq:entropyduality1} it follows that $\widetilde H^\downarrow_{1+s}(X_{A^n}|A'^nC)_\psi=n\log q-\bar H^\uparrow_{\nicefrac 1{1+s}}(Z_{A^n}|B)_\psi$, completing the proof.
\end{proof}

To establish Lemma~\ref{lem:oneshoterror}, observe that $P_{\textnormal{error}}(Z_{\hat A}|BZ_{\check A})_\psi$ is an average over $Z_{\check A}$ and therefore defines an average over a set of affine codes. The syndrome function of the codes is given by $\check f$. 
Then take $P$ to be the uniform distribution $Q$ and use \eqref{eq:renyireliability} to infer the existence of an affine code of size $q^m$ with 
\begin{equation}
-\log P_{\textnormal{error}}(W\circ \cC)_\psi \geq  
E_0(s, Q,W)-s\log q^m-s\log |\text{spec}(\psi_{A^nB})| -\log \tfrac1s\,.
\end{equation}

\section{Improved bounds for compression and privacy amplification}
\label{sec:compression}
Lemma~\ref{lem:basicerrorbound} immediately gives a lower bound on the error exponent of data compression with quantum side information, also known as information reconciliation, or quantum Slepian-Wolf coding. 
Recall that, for a given CQ state $\rho_{ZB}$, the goal is to compress $Z$ to a random variable $\check Z$ on a smaller alphabet such that $Z$ can be recovered from $\check Z$ along with the (quantum) side information stored in $B$ (e.g.\ by making a measurement on $B$). 
If we consider the i.i.d.\ case of $\rho_{ZB}^{\otimes n}$ and surjective linear functions to generate $\check Z$, then $Z^n$ is recoverable precisely when $\hat Z$ is. 
Meanwhile, the rate of the compression protocol is $R_{\textnormal{DC}}=\tfrac1n\log |\check A|$.  
Therefore, Lemma~\ref{lem:basicerrorbound} implies 
\begin{equation}
\label{eq:errorexpDC}
\lim_{n\to \infty}\tfrac{-1}n \log P_{\text{error}}(\hat Z|B^n\check Z)_{\rho^{\otimes n}} \geq \max_{\alpha \in [\nicefrac12,1]} \tfrac{1-\alpha}{\alpha}(R_{\textsc{dc}}-\bar H^\uparrow_{\alpha}(Z|B)_\rho)\,.
\end{equation}

The sphere-packing bound of Cheng et al.~\cite[Theorem 2]{cheng_non-asymptotic_2021} has very nearly the same form:
\begin{equation}
\label{eq:spherepackingDC}
\lim_{n\to \infty}\tfrac{-1}n \log P_{\text{error}}(\hat Z|B^n\check Z)_{\rho^{\otimes n}}\leq \sup_{\alpha \in [0,1]} \tfrac{1-\alpha}{\alpha}(R_{\textsc{dc}}-\bar H^\uparrow_{\alpha}(Z_A|B)_\rho)\,.
\end{equation}
Whenever the optimal $\alpha$ in the sphere-packing bound is at least one-half, the two bounds \eqref{eq:errorexpDC} and \eqref{eq:spherepackingDC} agree. From the analogous behavior for the bounds on the channel coding error exponent, we may surmise that this occurs for rates below a critical value, as at high enough rates zero-error compression potentially becomes possible. 

In fact, their Theorem 2 establishes not just an upper bound on the error exponent, but more generally a non-asymptotic lower bound on the error probability itself. 
Defining $E_{\textsc{sp}}(R)\coloneq \sup_{\alpha\in [0,1]}\tfrac{1-\alpha}\alpha(R-\bar H^\uparrow_\alpha(Z_A|A'B)_{\rho})$, they show that for large enough $n$ there is a constant $K$ such that 
\begin{equation}
\label{eq:spherepackingDC2}
-\tfrac{1}n \log P_{\text{error}}(\hat Z|B^n\check Z)_{\Psi'}\leq E_{\textsc{sp}}(R_{\textsc{dc}})+\tfrac12(1+|E_{\textsc{sp}}'(R_{\textsc{dc}})|)\tfrac{\log n}n+\tfrac Kn\,.
\end{equation}
Here $E'_{\textsc{sp}}$ denotes the derivative of the function $R\mapsto E_{\textsc{sp}}(R)$.  

This bound implies a non-asymptotic lower bound on the security parameter in privacy amplification, as we now show. 
Start with a general CQ state $\psi_{X_AC}=\sum_x P(x)\ketbra{\widetilde x}_A\otimes \theta_C(x)$ which specifies the input to the privacy amplification protocol. 
This state has a purification of the form  
\begin{equation}
\ket{\psi}_{AA'BC}=\sum_x \sqrt{P(x)}\ket{\widetilde x}_A\otimes \ket{\widetilde x}_{A'}\otimes \ket{\theta(x)}_{BC}\,,
\end{equation} 
where each $\ket{\theta(x)}_{BC}$ is a purification of the corresponding $\theta_{C}(x)$. 
Now consider the i.i.d.\ state $\ket{\Psi}_{A^nA'^nB^nC^n}=\ket{\psi}_{AA'BC}^{\otimes n}$. 
Choosing a surjective linear extractor function $\hat g$ and defining dual functions $\hat f$ and $\check f$, \eqref{eq:entropyduality2} implies (swapping $X_A\leftrightarrow Z_A$ and $B\leftrightarrow C$ therein) 
\begin{equation}
\max_{\sigma\in \cS_{C^n}} F(\Psi'_{X_{\hat A} C^n},\pi_{\hat A}\otimes \sigma_{C^n})^2 = P_{\textnormal{guess}}(Z_{\hat A}|A'^nB^n Z_{\check A})_\Psi\,.
\end{equation}
In terms of the purification distance $P(\rho,\sigma)\coloneq \sqrt{1-F(\rho,\sigma)^2}$ involving the actual marginal on $C^n$, this in turn implies 
\begin{equation}
P(\Psi_{X_{\hat A} C^n},\pi_{\hat X}\otimes \Psi_{C^n})^2 \geq P_{\text{error}}(Z_{\hat Z}|A'^nB^nZ_{\check A})_{\Psi}\,.
\end{equation}
Combining this bound with \eqref{eq:spherepackingDC2} gives
\begin{equation}
-\tfrac1n\log P(\Psi_{X_{\hat A} C^n},\pi_{\hat A}\otimes \Psi_{C^n})\leq  \tfrac12E_{\textsc{sp}}(R_{\textsc{dc}})+\tfrac14(1+|E_{\textsc{sp}}'(R_{\textsc{dc}})|)\tfrac{\log n}n+\tfrac K{2n}\,.
\end{equation}
We may express $E_{\textsc{SP}}(r)$ in terms of $\psi_{X_AC}$ using entropy duality as 
\begin{equation}
 E_{\textsc{SP}}(r)
 =\sup_{\alpha\geq 1}(\alpha-1)(r-\log q+\widetilde H^\downarrow_{\alpha}(X_A|C)_{\psi})
 \end{equation}
 and then define $E_{\textsc{SP-PA}}(r)\coloneq E_{\textsc{SP}}(\log q-r)$ so as to obtain 
\begin{equation}
E_{\textsc{SP-PA}}(R_{\textsc{pa}})=\sup_{\alpha\geq 1}(\alpha-1)(\widetilde H^\downarrow_{\alpha}(X_A|C)_{\psi'}-R_{\textsc{pa}})\,.
\end{equation}
Here we use the rate $R_{\textsc{pa}}\coloneq \tfrac1n\log |\hat A|$ of the privacy amplification protocol. 
Furthermore, $E'_{\textsc{SP-PA}}(r)=-E'_{\textsc{SP}}(\log q-r)$, and therefore we find (adjusting the constant $K$)
\begin{equation}
\label{eq:tighterpaupper}
-\tfrac1n\log P(\Psi_{X_{\hat A} C^n},\pi_{\hat A}\otimes \Psi_{C^n})\leq  \tfrac12E_{\textsc{sp-pa}}(R_{\textsc{pa}})+\tfrac14(1+|E_{\textsc{sp-pa}}'(R_{\textsc{pa}})|)\tfrac{\log n}n+\tfrac Kn\,.
\end{equation}
The first term in this expression gives the same $n\to\infty$ limit reported in \cite[Theorem 2]{li_tight_2023}; the additional terms give polynomial prefactors to the lower bound on the purification distance itself. Note that this bound is valid only for extractors based on surjective linear functions.

\section{Discussion}
We have established a lower bound to the CQ coding exponent which matches the form of the sphere-packing upper bound, just as in the classical case. 
This resolves the issue, which has been an open question since Burnashev and Holevo's initial investigations. 
However, the resolution does not directly proceed by analysis of a specific decoder, but instead takes a rather indirect route via duality arguments. 
It remains to find more direct coding theory arguments for the lower bound. 

\section*{Acknowledgments}
I thank Marco Tomamichel and Hao-Chung Cheng for useful discussions. 
This work was supported by the Swiss National Science Foundation through the Sinergia grant CRSII5\_186364, the National Center for Competence in Research for Quantum Science and Technology (QSIT), and the CHIST-ERA project No.\ 20CH21\_218782, as well as the Air Force Office of Scientific Research (AFOSR), grant FA9550-19-1-0202.

\section*{Related work}
After initial publication of this manuscript on the arXiv preprint server, the author learned of independent work by Li and Yang which takes a different approach to the problem using the method of types and permutation symmetry~\cite{li_reliability_2024}.

\appendix

\section{Duality statements}
\label{sec:duality}

We first establish two useful properties of conjugate bases. 
For fixed basis $\{\ket{z}\}$, let $U_{AA'}$ be the unitary with action $\ket{z}_A\otimes \ket{z'}_{A'}\mapsto \ket{z}_A\otimes \ket{z'+z}_{A'}$ and let $\Pi_{AA'}=\sum_{z} \ketbra{z}_A\otimes \ketbra{z}_{A'}$. 
Further, let $\mathcal S_{Z_AB}$ be the set of CQ density operators on $AB$ with $A$ classical in the $\ket{z}$ basis.
\begin{lemma}
Fix two orthonormal bases $\{\ket{z}\}_{z=0}^{d_A-1}$ and $\{\ket{\widetilde x}\}_{x=0}^{d_A-1}$ of $\mathcal H_A$ such that $|\braket{\widetilde x|z}|^2=\tfrac1{d_A}$ for all $x,z\in 0,\dots d_A-1$. Let $\mathcal P$ be the pinch map associated with the former and $\widetilde{\mathcal P}$ the pinch map associated with the latter. 
There exists a channel $\mathcal E_{AA'}$ such that for all $\theta_{AA'B}$ satisfying $\theta_{AA'B}=\Pi_{AA'}\theta_{AA'B}\Pi_{AA'}$, 
\begin{equation}
\label{eq:undoPx}
\mathcal E_{AA'}\circ\widetilde{\mathcal P}_A[\theta_{AA'B}]=\theta_{AA'B}\,.
\end{equation}
Furthermore, for any CQ state $\sigma_{A'B}=\cS_{Z_{A'}B}$, 
\begin{equation}
\label{eq:pxproj}
\widetilde{\cP}_{A}[\Pi_{AA'}(\id_A\otimes \sigma_{A'B})\Pi_{AA'}]=\pi_A\otimes \sigma_{A'B}\,.
\end{equation}
\end{lemma}
\begin{proof}
For the first statement, one choice for $\mathcal E_{AA'}$ is the channel with Kraus operators 
\begin{equation}
K_{AA'}(x)=\sqrt{d_A}\sum_{z} \braket{z|\widetilde x} (\ket{z}\bra{\widetilde x})_A\otimes (\ketbra{z})_{A'}\,.
\end{equation}
First let us confirm that these satisfy the normalization condition required for a quantum channel:
\begin{equation}
\begin{aligned}
\sum_x K_{AA'}(x)^* K_{AA'}(x)
&=d_A\sum_{xzz'}(\braket{\widetilde x|z'} (\ket{\widetilde x}\bra{z'})_A\otimes (\ketbra{z'})_{A'})(\braket{z|\widetilde x} (\ket{z}\bra{\widetilde x})_A\otimes (\ketbra{z})_{A'})\\
&=d_A\sum_{xz}|\braket{z|\widetilde x}|^2  \ketbra{\widetilde x}_A\otimes \ketbra{z}_A=\id_{AA'}\,.
\end{aligned}
\end{equation}
To see that the channel has the desired action, first observe that the Kraus operators of $\mathcal E_{AA'}\circ\widetilde P_A$ are again just $K_{AA'}(x)$. That is, $K_{AA'}(x)\ketbra{\widetilde x'}_A=\delta_{xx'}K_{AA'}(x)$. 
Furthermore, since the input state is assumed to be invariant under $\Pi_{AA'}$, we need only show that $K_{AA'}(x)\Pi_{AA'}$ is proportional to $\Pi_{AA'}$:
\begin{equation}
\begin{aligned}
K_{AA'}(x)\Pi_{AA'}
&=(\sqrt{d_A}\sum_{z'} \braket{z'|\widetilde x} (\ket{z'}\bra{\widetilde x})_A\otimes (\ketbra{z'})_{A'})(\sum_z \ketbra{z}_A\otimes \ketbra{z}_{A'})\\
&=\sqrt{d_A}\sum_z |\braket{z|\widetilde x}|^2\ket{z}\bra{z}_{A}\otimes \ketbra{z}_{A'}=\tfrac1{\sqrt{d_A}}\Pi_{AA'}\,.
\end{aligned}
\end{equation}

The second statement follows by direct calculation, using the fact that $\sigma_{A'B}=\sum_z \ketbra{z}_{A'}\otimes \varphi_B(z)$ for some set of subnormalized states $\varphi_B(z)$:
\begin{equation}
\begin{aligned}
\widetilde{\cP}_A[\Pi_{AA'}(\id_A\otimes \sigma_{A'B})\Pi_{AA'}]
&=\widetilde{\cP}_A[\sum_z \ketbra{z}_A\otimes \ketbra{z}_{A'}\otimes \varphi_B(z)]\\
&=\sum_{x,z} |\braket{z|\widetilde x}|^2\ketbra{\widetilde x}_A \otimes  \ketbra{z}_{A'}\otimes \varphi_B(z)
=\pi_A\otimes \sigma_{A'B}\,.
\end{aligned}
\end{equation}
This completes the proof.
\end{proof}

Following~\cite{coles_uncertainty_2012}, we consider a general relative entropy $\sD$ function on two positive semidefinite operators which satisfies the following three conditions for all positive semidefinite operators $\rho$ and $\sigma$:
\begin{enumerate}
  \item (\emph{Data processing}) For any quantum channel $\cN$, $\sD(\cN[\rho],\cN[\sigma])\leq \sD(\rho,\sigma)$, 
  \item (\emph{Null spaces}) For any positive semidefinite $\tau$, $\sD(\rho\oplus 0,\sigma\oplus \tau)=\sD(\rho,\sigma)$,
  \item (\emph{Normalization}) For any $c>0$, $\sD(\rho,c\sigma)=\sD(\rho,\sigma)+\log \tfrac1c$. 
\end{enumerate}
Using any such $\sD$, we define two conditional entropies $\sH^\uparrow$ and $\sH^\downarrow$ of bipartite quantum states by 
\begin{align}
\sH^\downarrow(A|B)_\rho &\coloneq -\sD(\rho_{AB},\id_A\otimes \rho_B)\,,\\
\sH^\uparrow(A|B)_\rho &\coloneq \max_{\sigma_B\in \cS_B} [-\sD(\rho_{AB},\id_A\otimes \sigma_B)]\,.
\end{align}
For each conditional entropy $\sH$, we define the dual $\widehat \sH$ by $\widehat H(A|B)_\rho\coloneq-\sH(A|C)_\rho$ for $\rho_{ABC}$ a pure state.

\begin{theorem}
For any state $\rho_{AC}$ let $\ket{\rho}_{ABC}$ be a purification and define $\ket{\psi}_{AA'BC}=U_{AA'}\ket{\rho}_{ABC}\ket{0}_{A'}$. 
Then for any two conjugate bases $X_A$ and $Z_A$ and any conditional entropy $\sH$  we have 
\begin{equation}
\widehat{\sH}(Z_A|C)_\rho+\sH(X_A|A'B)_{\psi}=\log d_A\,.
\end{equation}
\end{theorem}
\begin{proof}
For the case $\sH^{\downarrow}$, the theorem is proven by the following chain of equalities:
\begin{subequations}
\begin{align}
\widehat{\sH}(Z_A|C)_\rho
&=\widehat{\sH}(A|C)_{\mathcal P_A[\rho_{AC}]}\\
&=\widehat{\sH}(A|C)_\psi\\
&=-\sH(A|A'B)_\psi\\
&=\sD(\psi_{AA'B},\id_A\otimes \psi_{A'B})\\
&=\sD(\psi_{AA'B},\Pi_{AA'}(\id_A\otimes \psi_{A'B})\Pi_{AA'})\\
&=\sD(\widetilde{\cP}_A[\psi_{AA'B}],\widetilde{\cP}_A[\Pi_{AA'}(\id_A\otimes \psi_{A'B})\Pi_{AA'}])\\
&=\sD(\widetilde{\cP}_A[\psi_{AA'B}],\pi_A\otimes \psi_{A'B})\\
&=\log d_A-\sH(X_A|A'B)_\psi\,.
\end{align} 
\end{subequations}
The first equality is the definition of the conditional entropy of the measured state. 
Due to the form of $\ket{\psi}_{AA'BC}$, the marginal state $\psi_{AC}$ satisfies $\psi_{AC}=\cP_A[\rho_{AC}]$; therefore the second equality holds. 
The third is duality and the fourth the definition of the conditional entropy in terms of relative entropy. 
The fifth uses the null spaces property of the relative entropy. 
The two spaces in question are given by $\Pi_{AA'}$ and $\id_{AA'}-\Pi_{AA'}$. 
Note that $\psi'_{A'B}$ is a CQ state, and so $\id_A\otimes \psi'_{A'B}$ commutes with both projectors. 
Nominally the sixth should be an inequality by data processing, but equality holds due to~\eqref{eq:undoPx}. 
The seventh is \eqref{eq:pxproj} and the eighth and final equality is the scaling property and the definition of conditional entropy. 

For $\sH^\uparrow$ the argument is entirely similar. The fourth equality becomes 
\begin{equation}
-\sH^\uparrow(A|A'B)_\psi=\inf_{\sigma\in \cS_{Z_{A'}B}} \sD(\psi_{AA'B},\id_A\otimes \sigma_{A'B})\,,
\end{equation} and the remaining steps proceed as before, with the infimum used in the definition of the conditional entropy in the final step.  
\end{proof}

\section{One-shot privacy amplification bound}
\label{sec:hayashibound}
A family of (hash) functions $f:\cX\to \hat \cX$ is \emph{two-universal} when, for all $x_1,x_2\in \cX$, the probability under a uniformly-random choice of $f$ that $f(x_1)=f(x_2)$ even though $x_1\neq x_2$ is at most $1/|\hat \cX|$. 
In \cite[Appendix II]{hayashi_exponential_2011} it is shown that the family of functions defined by matrix action $x^n\to Gx^n$ for $G=\begin{pmatrix}\id & T\end{pmatrix}$ for $T$ a random Toeplitz matrix on $\df_q$ for $q$ prime is two-universal. 
\begin{theorem}[Theorem 1~\cite{hayashi_precise_2015}]
Consider a CQ state $\rho_{XE}$ and a two-universal family of hash functions $f:\cX\to \hat \cX$ with $M=|\cX|$. 
Then, for any $s\in [0,1]$, 
\begin{equation}
\mathbb E_f D(\rho_{\hat XE},\pi_{\hat X}\otimes \rho_E)\leq \tfrac1s|\textnormal{spec}(\rho_E)|^s |\hat\cX|^s2^{-s \widetilde H^\downarrow_{1+s}(X|E)_\rho}\,.
\end{equation}
\end{theorem}
\begin{proof}
Note that the bound is trivially true for $s=0$, and therefore we need only consider $s\in (0,1]$. 
Suppose $\rho_{XE}=\sum_x \ketbra{x}_X\otimes \varphi_E(x)$ and $\rho_{\hat XE}=\sum_y \ketbra{y}\otimes \sum_{x:f(x)=y}\varphi_E(x)$, for unnormalized states $\varphi_E(x)$. 
First we whittle the expression down a bit. 
\begin{subequations}
\begin{align}
\mathbb E_f D(\rho_{\hat XE},\pi_{\hat X}\otimes \rho_E)
&=\mathbb E_f D(\sum_{y=1}^M \ketbra{y}\otimes \sum_{f(x)=y} \varphi_E(x)\,,\, \pi_{\hat X}\otimes \rho_E) \\
&=\mathbb E_f \sum_{y=1}^M \tr[ \sum_{f(x)=y} \varphi_E(x) \big( \log \sum_{f(x')=y} \varphi_E(x')-\log \tfrac 1M\rho_E\big)]\\
&=\mathbb E_f \sum_x \tr[\varphi_E(x) \big( \log \sum_{x':f(x')=f(x)} \varphi_E(x')-\log \tfrac 1M\rho_E\big)]\\
&=-\tfrac1M\tr[\rho_E \log\rho_E] + \mathbb E_f \sum_x \tr[\varphi_E(x) \log \sum_{x':f(x')=f(x)} \varphi_E(x')]
\end{align}
\end{subequations}
Then we use convexity of the logarithm to move the expectation inside the log in the second term:
\begin{subequations}
\begin{align}
\mathbb E_f &\sum_x \tr[\varphi_E(x) \log \sum_{x':f(x')=f(x)} \varphi_E(x')]\nonumber\\
&\leq \sum_x \tr[\varphi_E(x) \log \big(\mathbb E_f \sum_{x':f(x')=f(x)} \varphi_E(x')\big)]\\
&=\sum_x \tr[\varphi_E(x) \log \big(\varphi_E(x)+ \mathbb E_f \sum_{x'\neq x:f(x')=f(x)} \varphi_E(x')\big)]\\
&\leq \tr[\sum_x \varphi_E(x)\log \big(\varphi_E(x)+ \tfrac1M \sum_{x'\neq x} \varphi_E(x)\big)]\\
&\leq \tr[\sum_x \varphi_E(x)\log\big( \varphi_E(x)+ \tfrac1M \varphi_E\big)]\,.
\end{align}
\end{subequations}
The first inequality is convexity of log, the second is the 2-universal condition on the family of functions and monotonicity of log, and the final inequality is again monotonicity of log. 
Altogether we have 
\begin{equation}
\mathbb E_f D(\rho_{\hat XE},\pi_{\hat X}\otimes \rho_E) \leq \tr[\sum_x \varphi_E(x)\Big(\log\big( \varphi_E(x)+ \tfrac1M \varphi_E\big)-\log \tfrac1M \varphi_E\Big)]\,.
\end{equation}

Now the problem is that we cannot combine the $\log$ terms, since $\varphi_E$ does not necessarily commute with $\varphi_E(x)$.
So we pinch $\varphi_E(x)$ in the basis of $\varphi_E$. 
Denoting the pinched states with a bar and $\nu_E=|\textnormal{spec}(\rho_E)|$, this step gives  
\begin{equation}
\log\big( \varphi_E(x)+ \tfrac1M \varphi_E\big)
\leq \log\big( v_E \bar{\varphi}_E(x)+ \tfrac1M \varphi_E\big)\,,
\end{equation}
and therefore
\begin{equation}
\mathbb E_f D(\rho_{\hat XE},\pi_{\hat X}\otimes \rho_E) \leq \tr[\sum_x \varphi_E(x)\log\big( M v \bar\varphi_E(x)\varphi_E^{-1}+ \id_E\big)]\,.
\end{equation}

Next make use of the fact that $\log (\id+X)\leq \tfrac1sX^s$ for $s\in (0,1]$ and nonnegative $X$~\cite[Lemma 5]{hayashi_precise_2015} to get  
\begin{subequations}
\begin{align}
\mathbb E_f D(\rho_{\hat XE},\pi_{\hat X}\otimes \rho_E) 
&\leq \frac{v^sM^s}s\tr[\sum_x \varphi_E(x) \bar\varphi_E(x)^s\varphi_E^{-s}]\\
&= \frac{v^sM^s}s\tr[\sum_x \bar\varphi_E(x)^{1+s}\varphi_E^{-s}]\\
&=\frac{v^sM^s}s 2^{s \widetilde D_{1+s}(\bar \rho_{XE},\id_X\otimes \rho_E)}\\
&=\frac{v^sM^s}s 2^{-s \widetilde H_{1+s}^\downarrow (X|E)_{\bar{\rho}_{XE}}}\\
&\leq \frac{v^sM^s}s 2^{-s \widetilde H_{1+s}^\downarrow (X|E)_{{\rho}_{XE}}}\,.
\end{align} 
\end{subequations}
The inequality in the final step is data processing. 
\end{proof}

\printbibliography[heading=bibintoc,title={\large References}]

\end{document}